\documentclass[twoside, journal, onecolumn, double, 11pt]{IEEEtran}

\usepackage{epsfig}
\usepackage{graphicx}
\usepackage{amsmath}
\usepackage{amssymb}
\usepackage{float}
\usepackage{url}
\newtheorem{definition}{Definition}

\newtheorem{corollary}{Corollary}
\newtheorem{theorem}{Theorem}

\newtheorem{remark}{Remark}

\newcommand{\dsum}{\displaystyle\sum}

\newcommand{\naturals}{\ensuremath{\mathbb{N}}}
\newcommand{\reals}{\ensuremath{\mathbb{R}}}

\newcommand{\expectation}{\ensuremath{\mathbb{E}}}

\begin{document}

\title{\Large{Improved Lower Bounds on the Total Variation Distance for the Poisson Approximation}}

\markboth{Final version. Last updated: July~16, 2013. To appear in the Statistics and Probability
Letters}{I. SASON: Improved Lower Bounds on the Total Variation Distance for the Poisson Approximation}

\author{\IEEEauthorblockN{Igal Sason\\
\hspace*{-0.4cm} Department of Electrical Engineering\\
Technion, Haifa 32000, Israel\\
E-mail: \url{sason@ee.technion.ac.il}\\}}

\maketitle 

\begin{abstract}
New lower bounds on the total variation distance between the distribution of a sum
of independent Bernoulli random variables and the Poisson random variable (with the
same mean) are derived via the Chen-Stein method. The new bounds rely on a non-trivial
modification of the analysis by Barbour and Hall (1984) which surprisingly gives a
significant improvement. A use of the new lower bounds is addressed.
\end{abstract}


\section{Introduction}
Convergence to the Poisson distribution, for the number of occurrences of possibly
dependent events, naturally arises in various applications. Following the work of
Poisson, there has been considerable interest in how well the Poisson distribution
approximates the binomial distribution.

The basic idea which serves for the starting point of the so called {\em Chen-Stein method
for the Poisson approximation} is the following (see Chen (1975)). Let $\{X_i\}_{i=1}^n$ be
independent Bernoulli random variables with $\expectation(X_i) = p_i$.
Let $W \triangleq \sum_{i=1}^n X_i$ and
$V_i \triangleq \sum_{j \neq i} X_j$ for every $i \in \{1, \ldots, n\}$, and
$Z \sim \text{Po}(\lambda)$ with mean $\lambda \triangleq \sum_{i=1}^n p_i$.
It is easy to show that
\begin{equation}
\expectation[\lambda f(Z+1) - Z f(Z)] = 0
\label{eq: basic equation of the Chen-Stein method for Poisson approximation}
\end{equation}
holds for an arbitrary bounded function $f: \naturals_0 \rightarrow \reals$
where $\naturals_0 \triangleq \{0, 1, \ldots \}$.
Furthermore (see, e.g., Chapter~2 in Ross and Pek{\"{o}}z (2007))
\begin{equation}
\expectation\bigl[\lambda f(W+1) - W f(W)\bigr]
= \sum_{j=1}^n p_j^2 \, \expectation\bigl[f(V_j+2) - f(V_j+1) \bigr]
\label{eq: first step in the derivation of the improved lower bound on total variation distance}
\end{equation}
which then serves to provide rigorous bounds on the difference between the distributions
of $W$ and $Z$, by the Chen-Stein method for Poisson approximations.
This method, and more generally the so called {\em Stein method}, serves as a powerful
tool for the derivation of rigorous bounds for various distributional approximations. Nice
expositions of this method are provided by, e.g., Arratia et al. (1990),
Ross and Pek{\"{o}}z (2007) and Ross (2011). Furthermore, some interesting
links between the Chen-Stein method and information-theoretic functionals in the
context of Poisson and compound Poisson approximations are
provided by Barbour et al. (2010).

Throughout this letter, the term `distribution' refers to a discrete
probability mass function of an integer-valued random variable. In the following,
we introduce some known results that are related to the presentation of the
new results.
\begin{definition}
Let $P$ and $Q$ be two probability measures defined on a set $\mathcal{X}$.
Then, the total variation distance between $P$ and $Q$ is defined by
\begin{equation}
d_{\text{TV}}(P, Q)  \triangleq \sup_{\text{Borel} \,
A \subseteq \mathcal{X}} \bigl( P(A) - Q(A) \bigr)
\label{eq: total variation distance}
\end{equation}
where the supremum is taken w.r.t. all the Borel subsets $A$ of
$\mathcal{X}$.
If $\mathcal{X}$ is a countable set then \eqref{eq: total variation distance}
is simplified to
\begin{equation}
d_{\text{TV}}(P, Q) = \frac{1}{2} \sum_{x \in \mathcal{X}} |P(x) - Q(x)| =
\frac{||P-Q||_1}{2}
\label{eq: the L1 distance is twice the total variation distance}
\end{equation}
so the total variation distance is equal to half of the $L_1$-distance
between the two probability distributions.
\label{definition: total variation distance}
\end{definition}

Among old and interesting
results that are related to the Poisson approximation, Le Cam's inequality (see Le Cam (1960)) provides
an upper bound on the total variation distance between the distribution of
the sum $W = \sum_{i=1}^n X_i$ of $n$ independent Bernoulli random variables $\{X_i\}_{i=1}^n$,
where $X_i \sim \text{Bern}(p_i)$, and a Poisson distribution $\text{Po}(\lambda)$ with
mean $\lambda = \sum_{i=1}^n p_i$. This inequality states that
$
d_{\text{TV}}\bigl(P_{W}, \text{Po}(\lambda)\bigr) \leq \sum_{i=1}^n p_i^2
$
so if, e.g., $X_i \sim \text{Bern}\bigl(\frac{\lambda}{n}\bigr)$ for every $i \in \{1, \ldots, n\}$
(referring to the case where $W$ is binomially distributed) then this upper bound is equal to
$\frac{\lambda^2}{n}$, decaying to zero as
$n \rightarrow \infty$.
The following theorem combines Theorems~1 and 2 of Barbour and Hall (1984),
and its proof relies on the Chen-Stein method:

\begin{theorem}
Let $W = \sum_{i=1}^n X_i$ be a sum of $n$ independent Bernoulli random
variables with $\expectation(X_i) = p_i$ for $i \in \{1, \ldots, n\}$,
and $\expectation(W) = \lambda$. Then, the total variation distance
between the probability distribution of $W$ and the Poisson
distribution with mean $\lambda$ satisfies
\begin{equation}
\frac{1}{32} \, \Bigl(1 \wedge \frac{1}{\lambda}\Bigr) \, \sum_{i=1}^n p_i^2
\leq d_{\text{TV}}(P_W, \text{Po}(\lambda)) \leq
\left(\frac{1-e^{-\lambda}}{\lambda}\right) \, \sum_{i=1}^n p_i^2
\label{eq: bounds on the total variation distance - Barbour and Hall 1984}
\end{equation}
where $a \wedge b \triangleq \min\{a,b\}$ for every $a, b \in \reals$.
\label{theorem: bounds on the total variation distance - Barbour and Hall 1984}
\end{theorem}

As a consequence of Theorem~\ref{theorem: bounds on the total variation distance - Barbour and Hall 1984},
it follows that the ratio between the upper and lower bounds in
\eqref{eq: bounds on the total variation distance - Barbour and Hall 1984}
is not larger than~32, irrespectively of the values of $\{p_i\}$. The factor
$\frac{1}{32}$ in the lower bound was claimed to be improvable to $\frac{1}{14}$
with no explicit proof (see Remark~3.2.2 in Barbour et al. (1992)).
This shows that,
for independent Bernoulli random variables, these bounds are essentially tight.
Furthermore, note that the upper bound in
\eqref{eq: bounds on the total variation distance - Barbour and Hall 1984} improves
Le Cam's inequality; for large values of $\lambda$, this improvement is by approximately
a factor of $\frac{1}{\lambda}$.

This letter presents new lower bounds on the total variation distance between the
distribution of a sum of independent Bernoulli random variables and the Poisson random variable
(with the same mean). The derivation of these new bounds generalizes
and improves the analysis by Barbour and Hall (1984), based on the Chen-Stein method
for the Poisson approximation. This letter concludes by outlining a use of the new lower bounds
for the analysis in Sason (2012), followed by a comparison of
the new bounds to previously reported bounds.

This work forms a continuation of the line of work in
Barbour and Chen (2005)--Kontoyiannis et al. (2005)
where the Chen-Stein method was studied in the context of the Poisson and compound
Poisson approximations,
and it was linked to an information-theoretic context by Barbour et al. (2010),
Kontoyiannis et al. (2005), and Sason (2012).

\section{Improved Lower Bounds on the Total Variation Distance}
\label{section: Improved lower bounds on the total variation distance}
In the following, we introduce an improved lower bound on the total variation distance
and then provide a loosened version of this bound that is expressed in closed form.
\begin{theorem}
In the setting of
Theorem~\ref{theorem: bounds on the total variation distance - Barbour and Hall 1984},
the total variation distance between the probability distribution of $W$ and the Poisson
distribution with mean $\lambda$ satisfies the inequality
\begin{equation}
K_1(\lambda) \, \sum_{i=1}^n p_i^2
\leq d_{\text{TV}}(P_W, \text{Po}(\lambda)) \leq
\left(\frac{1-e^{-\lambda}}{\lambda}\right) \sum_{i=1}^n p_i^2
\label{eq: improved lower bound on the total variation distance}
\end{equation}
where
\vspace*{-0.3cm}
\begin{equation}
K_1(\lambda) \triangleq \sup_{\small \begin{array}{ll}
& \alpha_1, \alpha_2 \in \reals, \\
& \alpha_2 \leq \lambda + \frac{3}{2}, \\
& \hspace*{0.4cm} \theta > 0 \\
\end{array}} \left( \frac{1-h_{\lambda}(\alpha_1, \alpha_2, \theta)}{2
\, g_{\lambda}(\alpha_1, \alpha_2, \theta)} \right)
\label{eq: K1 in the lower bound on the total variation distance}
\end{equation}
and
\vspace*{-0.5cm}
\begin{eqnarray}
&& \hspace*{-0.5cm} h_{\lambda}(\alpha_1, \alpha_2, \theta) \triangleq
\frac{3 \lambda + (2-\alpha_2+\lambda)^3 - (1-\alpha_2+\lambda)^3}{\theta \lambda} \nonumber \\
&& \hspace*{2cm} +\frac{|\alpha_1 - \alpha_2| \, \bigl(2 \lambda + |3-2\alpha_2| \bigr)
\, \exp\left(-\frac{(1-\alpha_2)_+^2}{\theta \lambda}\right)}{\theta \lambda}
\label{eq: h in the lower bound on the total variation distance} \\[0.3cm]
&& \hspace*{-0.5cm} x_+ \triangleq \max\{x, 0\}, \quad
x_+^2 \triangleq \bigl(x_+)^2, \quad
\forall \, x \in \reals \\[0.1cm]
&& \hspace*{-0.5cm} g_{\lambda}(\alpha_1, \alpha_2, \theta) \triangleq
\max \left\{ \, \left| \left(1 + \sqrt{\frac{2}{\theta \lambda e}} \cdot
|\alpha_1 - \alpha_2|\right) \lambda + \max_{u_i} \bigl\{ x(u_i) \bigr\} \right|, \right. \nonumber \\[0.3cm]
&& \hspace*{3.2cm} \left. \left| \left(2 e^{-\frac{3}{2}} + \sqrt{\frac{2}{\theta \lambda e}}  \cdot
|\alpha_1 - \alpha_2|\right) \lambda - \min_{u_i} \bigl\{x(u_i)\bigr\} \right| \,
\right\} \label{eq: g in the lower bound on the total variation distance} \\[0.1cm]
&& \hspace*{-0.5cm} x(u) \triangleq (c_0 + c_1 u + c_2 u^2) \, \exp(-u^2),
\quad \forall \, u \in \reals \label{eq: function x} \\[0.2cm]
&& \hspace*{-0.5cm} \{u_i\} \triangleq \Bigl\{ u \in \reals:
\, 2 c_2 u^3 + 2 c_1 u^2 - 2(c_2 - c_0) u - c_1 = 0\Bigr\}
\label{eq: zeros of a cubic polynomial equation} \\[0.1cm]
&& \hspace*{-0.5cm} c_0 \triangleq (\alpha_2 - \alpha_1) (\lambda - \alpha_2)
\label{eq: c0} \\
&& \hspace*{-0.5cm} c_1 \triangleq \sqrt{\theta \lambda}
\, (\lambda + \alpha_1 - 2 \alpha_2) \label{eq: c1} \\
&& \hspace*{-0.5cm} c_2 \triangleq -\theta \lambda.  \label{eq: c2}
\end{eqnarray}
\label{theorem: improved lower bound on the total variation distance}
\end{theorem}
\begin{proof}
See Section~\ref{subsubsection: Proof of the theorem with the improved lower bound on the total variation distance}. The derivation relies on the Chen-Stein method
for the Poisson approximation, and it improves (significantly) the constant in
the lower bound of Theorem~2 of Barbour and Hall (1984).
\end{proof}

\begin{remark}
The upper and lower bounds on the total variation distance in
\eqref{eq: improved lower bound on the total variation distance}
scale like $\sum_{i=1}^n p_i^2$, similarly to the known bounds in
Theorem~\ref{theorem: bounds on the total variation distance -
Barbour and Hall 1984}, but they offer a significant improvement
in their tightness (see
Section~\ref{section: Connection of the improved lower bounds with known results}).
\end{remark}

\vspace*{0.1cm}
\begin{remark}
The cardinality of the set $\{u_i\}$
in \eqref{eq: zeros of a cubic polynomial equation} is equal to~3
(see Section~\ref{subsubsection: Proof of the theorem with the improved lower bound on the total variation distance}).
\label{remark: the size of the set of real zeros is equal to 3}
\end{remark}

\vspace*{0.1cm}
\begin{remark}
The optimization that is required for the computation of $K_1$ in
\eqref{eq: K1 in the lower bound on the total variation distance}
w.r.t. the three parameters $\alpha_1, \alpha_2 \in \reals$ and $\theta \in \reals^+$
is performed numerically.
\end{remark}

\vspace*{0.1cm}
In the following, we introduce a looser lower bound on the total variation
distance as compared to the lower bound in
Theorem~\ref{theorem: improved lower bound on the total variation distance},
but its advantage is that it is expressed in closed form.
Both lower bounds improve (significantly) the lower bound in Theorem~2 of Barbour and Hall (1984).
The following lower bound follows from Theorem~\ref{theorem: improved lower bound on the total variation distance}
via the special choice of $\alpha_1 = \alpha_2 = \lambda$ that is included in the
optimization set for $K_1$ on the right-hand side of
\eqref{eq: K1 in the lower bound on the total variation distance}.
Following this sub-optimal choice, the lower bound in the next corollary
is obtained by a derivation of a closed-form expression for the third free
parameter $\theta \in \reals^+$ (in fact, this was our first step towards
the derivation of an improved lower bound on the total variation distance).

\vspace*{0.1cm}
\begin{corollary}
Under the assumptions in Theorem~\ref{theorem: improved lower bound on the total variation distance},
then
\begin{equation}
\widetilde{K}_1(\lambda) \, \sum_{i=1}^n p_i^2
\leq d_{\text{TV}}(P_W, \text{Po}(\lambda)) \leq
\left(\frac{1-e^{-\lambda}}{\lambda}\right) \sum_{i=1}^n p_i^2
\label{eq: first improvement of the lower bound on the total variation distance}
\end{equation}
where
\begin{eqnarray}
&& \hspace*{-1cm} \widetilde{K}_1(\lambda) \triangleq \frac{e}{2 \lambda} \; \frac{1 - \frac{1}{\theta}
\, \left(3+\frac{7}{\lambda}\right)}{\theta + 2 e^{-1/2}}
\label{eq: tilted K_1 in the lower bound on the total variation distance} \\[0.3cm]
&& \hspace*{-1cm} \theta \triangleq 3 + \frac{7}{\lambda} + \frac{1}{\lambda} \cdot
\sqrt{(3\lambda+7)\bigl[(3+2e^{-1/2}) \lambda + 7\bigr]}.
\label{eq: optimal theta for alpha1 and alpha2 equal to lambda}
\end{eqnarray}
\label{corollary: lower bound on the total variation distance}
\end{corollary}
\begin{proof}
See Section~\ref{subsubsection: Proof of the corollary with the improved lower bound on the total variation distance}.
\end{proof}

\section{Outlook}
\label{section: Applications}
We conclude our discussion in this letter by outlining a use of the new lower bounds
in this work: the use of the new lower bound on the total variation distance for the Poisson
approximation of a sum of independent Bernoulli random variables is
exemplified by Sason (2012). This work introduces new
entropy bounds for discrete random variables via maximal coupling, providing
bounds on the difference between the entropies of two discrete random variables
in terms of the local and total variation distances between their probability mass functions.
The new lower bound on the total variation distance for the Poisson approximation
from this work was involved in the calculation of some improved bounds on the difference
between the entropy of a sum of independent Bernoulli random variables and the entropy
of a Poisson random variable of the same mean. A possible application of the latter
problem is related to getting bounds on the sum-rate capacity of a noiseless
$K$-user binary adder multiple-access channel (see Sason (2012)).

\section{Proofs of the New Bounds}
\label{section: proofs}

\subsection{Proof of Theorem~\ref{theorem: improved lower bound on the
total variation distance}}
\label{subsubsection: Proof of the theorem with the improved lower bound on the total
variation distance}
The proof of Theorem~\ref{theorem: improved lower bound on the total variation distance}
starts similarly to the proof of Theorem~2 of Barbour and Hall (1984). However, it significantly
deviates from the original analysis in order to derive an improved lower bound on the
total variation distance.

Let $\{X_i\}_{i=1}^n$ be independent Bernoulli random variables with $\expectation(X_i) = p_i$.
Let $W \triangleq \sum_{i=1}^n X_i$,
$V_i \triangleq \sum_{j \neq i} X_j$ for every $i \in \{1, \ldots, n\}$, and
$Z \sim \text{Po}(\lambda)$ with mean $\lambda \triangleq \sum_{i=1}^n p_i$.
From the basic equation of the Chen-Stein method, equation
\eqref{eq: basic equation of the Chen-Stein method for Poisson approximation}
holds for an arbitrary bounded function $f: \naturals_0 \rightarrow \reals$.
Furthermore, it follows from the proof of Theorem~2 of Barbour and Hall (1984) that
\begin{eqnarray}
d_{\text{TV}}(P_W, \, \text{Po}(\lambda)) \geq
\frac{ \dsum_{j=1}^n \Bigl\{p_j^2 \, \expectation\bigl[f(V_j+2) - f(V_j+1) \bigr] \Bigr\}}{2 \,
\sup_{k \in \naturals_0} \bigl| \lambda f(k+1) - k f(k) \bigr|}
\label{eq: fourth step in the derivation of the improved lower bound on total variation distance}
\end{eqnarray}
which holds, in general, for an arbitrary bounded function $f: \naturals_0 \rightarrow \reals$.

At this point, we deviate from the proof of Theorem~2 of Barbour and Hall (1984) by generalizing and
refining (in a non-trivial way) the original analysis. The general problem
with the current lower bound in \eqref{eq: fourth step in the derivation of the improved lower
bound on total variation distance} is that it is not calculable in closed form for a given $f$, so one needs
to choose a proper function $f$ and derive a closed-form expression for a lower bound on the right-hand side of
\eqref{eq: fourth step in the derivation of the improved lower bound on total variation distance}.
To this end, let
\begin{equation}
f(k) \triangleq (k-\alpha_1) \, \exp\biggl(-\frac{(k-\alpha_2)^2}{\theta \lambda}\biggr) \, ,
\quad \forall \, k \in \naturals_0
\label{eq: proposed f for the derivation of the improved lower bound on the total variation
distance}
\end{equation}
where $\alpha_1, \alpha_2 \in \reals$ and $\theta \in \reals^+$ are fixed constants (note that
$\theta$ in \eqref{eq: proposed f for the derivation of the improved lower bound on the total variation distance}
needs to be positive for $f$ to be a bounded function). In order
to derive a lower bound on the total variation distance, we calculate a lower bound on the
numerator and an upper bound on the denominator of the right-hand side
of \eqref{eq: fourth step in the derivation of the improved lower bound on total variation distance}
for the function $f$ in \eqref{eq: proposed f for the derivation of the improved lower bound on the
total variation distance}. Referring to the numerator of the right-hand side
of \eqref{eq: fourth step in the derivation of the improved lower bound on total variation
distance} with $f$ in \eqref{eq: proposed f for the derivation of the improved lower bound on the
total variation distance}, for every $j \in \{1, \ldots, n\}$,
\begin{eqnarray}
&& f(V_j+2) - f(V_j+1) \nonumber \\[0.1cm]
&& = \int_{V_j + 1 - \alpha_2}^{V_j + 2 - \alpha_2} \frac{\mathrm{d}}{\mathrm{d}u}
\left( (u+\alpha_2-\alpha_1) \, \exp\Bigl(-\frac{u^2}{\theta \lambda}\Bigr) \right) \,
\mathrm{d}u  \nonumber \\[0.2cm]
&& = \int_{V_j + 1 - \alpha_2}^{V_j + 2 - \alpha_2} \left(1 - \frac{2u
(u+\alpha_2-\alpha_1)}{\theta \lambda} \right) \exp\Bigl(-\frac{u^2}{\theta \lambda}\Bigr) \,
\mathrm{d}u \nonumber \\[0.2cm]
&& = \int_{V_j + 1 - \alpha_2}^{V_j + 2 - \alpha_2} \left(1 - \frac{2 u^2}{\theta \lambda}\right)
\, \exp\Bigl(-\frac{u^2}{\theta \lambda}\Bigr) \, \mathrm{d}u -
\frac{2(\alpha_2 - \alpha_1)}{\theta \lambda} \int_{V_j + 1 - \alpha_2}^{V_j + 2 - \alpha_2}
u \, \exp\Bigl(-\frac{u^2}{\theta \lambda}\Bigr) \, \mathrm{d}u \nonumber \\[0.2cm]
&& = \int_{V_j + 1 - \alpha_2}^{V_j + 2 - \alpha_2} \left(1 - \frac{2 u^2}{\theta \lambda}\right)
\, \exp\Bigl(-\frac{u^2}{\theta \lambda}\Bigr) \, \mathrm{d}u \nonumber \\
&& \hspace*{0.4cm} - (\alpha_2 - \alpha_1) \left[ \exp\biggl(-\frac{(V_j + 2 - \alpha_2)^2}{\theta
\lambda}\biggr) - \exp\biggl(-\frac{(V_j + 1 - \alpha_2)^2}{\theta \lambda}\biggr) \right].
\label{eq: fifth step in the derivation of the improved lower bound on total variation distance}
\end{eqnarray}
We rely in the following on the inequality
$$(1-2x) \, e^{-x} \geq 1-3x, \quad \forall \, x \geq 0.$$
Applying it to the integral on the right-hand side of
\eqref{eq: fifth step in the derivation of the improved lower bound on total variation distance}
gives that
\begin{eqnarray}
&& f(V_j+2) - f(V_j+1) \nonumber \\[0.1cm]
&& \geq \int_{V_j + 1 - \alpha_2}^{V_j + 2 - \alpha_2} \left(1 - \frac{3 u^2}{\theta
\lambda}\right) \, \mathrm{d}u -
(\alpha_2 - \alpha_1) \left[ \exp\biggl(-\frac{(V_j + 2 - \alpha_2)^2}{\theta \lambda}\biggr) -
\exp\biggl(-\frac{(V_j + 1 - \alpha_2)^2}{\theta \lambda}\biggr) \right] \nonumber \\[0.1cm]
&& \geq 1 - \frac{\bigl(V_j + 2 - \alpha_2\bigr)^3 - \bigl(V_j + 1 - \alpha_2\bigr)^3}{\theta
\lambda} \nonumber \\[0.1cm]
&& \hspace*{0.4cm} - \bigl|\alpha_2 - \alpha_1\bigr| \cdot
\left| \exp\biggl(-\frac{(V_j + 2 - \alpha_2)^2}{\theta \lambda}\biggr) -
\exp\biggl(-\frac{(V_j + 1 - \alpha_2)^2}{\theta \lambda}\biggr) \right|.
\label{eq: sixth step in the derivation of the improved lower bound on total variation distance}
\end{eqnarray}
In order to proceed, note that if $x_1, x_2 \geq 0$ then (on the basis of the mean-value theorem
of calculus)
\begin{eqnarray*}
&& | e^{-x_2} - e^{-x_1} | \\
&& = \bigl|e^{-c} \, (x_1 - x_2)\bigr|   \quad \mbox{for some} \; \; c \in [x_1, x_2] \\[0.1cm]
&& \leq e^{-\min\{x_1, x_2\}} \, |x_1 - x_2|
\end{eqnarray*}
which, by applying it to the second term on the right-hand side of
\eqref{eq: sixth step in the derivation of the improved lower bound on total variation distance},
gives that for every $j \in \{1, \ldots, n\}$
\begin{eqnarray}
&&\left| \exp\biggl(-\frac{(V_j + 2 - \alpha_2)^2}{\theta \lambda}\biggr) -
\exp\biggl(-\frac{(V_j + 1 - \alpha_2)^2}{\theta \lambda}\biggr) \right|
\nonumber \\[0.1cm]
&& \leq \exp\left(-\frac{\min\Bigl\{(V_j+2-\alpha_2)^2, \, (V_j+1-\alpha_2)^2
\Bigr\}}{\theta \lambda}\right) \cdot
\left(\frac{(V_j+2-\alpha_2)^2-(V_j+1-\alpha_2)^2}{\theta \lambda}\right) \, .
\label{eq: seventh step in the derivation of the improved lower bound on total variation distance}
\end{eqnarray}
Since $V_j = \sum_{i \neq j} X_i \geq 0$ then
\begin{eqnarray}
&& \min\Bigl\{(V_j+2-\alpha_2)^2, \, (V_j+1-\alpha_2)^2\Bigr\} \nonumber \\[0.1cm]
&& \geq \left\{ \begin{array}{cl}
0 \quad & \mbox{if $\alpha_2 \geq 1$} \\[0.1cm]
(1-\alpha_2)^2 \quad & \mbox{if $\alpha_2 < 1$}
\end{array}
\right. \nonumber \\[0.1cm]
&& = \bigl(1-\alpha_2\bigr)_{+}^2
\label{eq: 8th step in the derivation of the improved lower bound on total variation distance}
\end{eqnarray}
where $$x_+ \triangleq \max\{x,0\}, \quad x_+^2 \triangleq \bigl(x_+ \bigr)^2, \quad
\forall \, x \in \reals.$$ Hence, the combination
of the two inequalities in
\eqref{eq: seventh step in the derivation of the improved lower bound
on total variation distance}--\eqref{eq: 8th step in the derivation
of the improved lower bound on total variation distance} gives that
\begin{eqnarray}
&&\left| \exp\biggl(-\frac{(V_j + 2 - \alpha_2)^2}{\theta \lambda}\biggr) -
\exp\biggl(-\frac{(V_j + 1 - \alpha_2)^2}{\theta \lambda}\biggr) \right|
\nonumber \\[0.1cm]
&& \leq \exp\left(-\frac{(1-\alpha_2)_+^2}{\theta \lambda} \right) \cdot
\left(\frac{\left|(V_j + 2 - \alpha_2)^2 -
(V_j + 1 - \alpha_2)^2 \right|}{\theta \lambda} \right) \nonumber \\[0.1cm]
&& = \exp\left(-\frac{(1-\alpha_2)_+^2}{\theta \lambda} \right) \cdot
\frac{\left|2V_j + 3 - 2 \alpha_2 \right|}{\theta \lambda}
\nonumber \\[0.1cm]
&& \leq \exp\left(-\frac{(1-\alpha_2)_+^2}{\theta \lambda} \right) \cdot
\frac{2V_j + \left|3 - 2 \alpha_2 \right|}{\theta \lambda}
\label{eq: ninth step in the derivation of the improved lower bound on total variation distance}
\end{eqnarray}
and therefore, a combination of the inequalities in
\eqref{eq: sixth step in the derivation of the improved lower bound on total variation distance}
and
\eqref{eq: ninth step in the derivation of the improved lower bound on total variation distance}
gives that
\begin{eqnarray}
&& f(V_j+2) - f(V_j+1) \nonumber \\[0.1cm]
&& \geq 1 - \frac{\bigl(V_j + 2 - \alpha_2\bigr)^3 - \bigl(V_j + 1 - \alpha_2\bigr)^3}{\theta
\lambda} \nonumber \\[0.1cm]
&& \hspace*{0.4cm} - \bigl|\alpha_2 - \alpha_1\bigr| \cdot
\exp\left(-\frac{(1-\alpha_2)_+^2}{\theta \lambda} \right) \cdot
\frac{2V_j + \left|3 - 2 \alpha_2 \right|}{\theta \lambda} \; .
\label{eq: 10th step in the derivation of the improved lower bound on total variation distance}
\end{eqnarray}
Let $U_j \triangleq V_j - \lambda$; then
\begin{eqnarray}
&& f(V_j+2) - f(V_j+1) \nonumber \\[0.1cm]
&& \geq 1 - \frac{\bigl(U_j + \lambda + 2 - \alpha_2\bigr)^3 -
\bigl(U_j + \lambda + 1 - \alpha_2\bigr)^3}{\theta
\lambda} \nonumber \\[0.1cm]
&& \hspace*{0.4cm} - \bigl|\alpha_2 - \alpha_1\bigr| \cdot
\exp\left(-\frac{(1-\alpha_2)_+^2}{\theta \lambda} \right) \cdot
\frac{2U_j + 2\lambda + \left|3 - 2 \alpha_2 \right|}{\theta \lambda} \nonumber \\[0.1cm]
&& = 1 - \frac{3 U_j^2 + 3 \bigl(3-2\alpha_2 + 2 \lambda\bigr) U_j +
(2-\alpha_2+\lambda)^3 - (1-\alpha_2+\lambda)^3}{\theta \lambda} \nonumber \\[0.1cm]
&& \hspace*{0.4cm} - \bigl|\alpha_2 - \alpha_1\bigr| \cdot
\exp\left(-\frac{(1-\alpha_2)_+^2}{\theta \lambda} \right) \cdot
\frac{2U_j + 2\lambda + \left|3 - 2 \alpha_2 \right|}{\theta \lambda} \; .
\label{eq: 11th step in the derivation of the improved lower bound on total variation distance}
\end{eqnarray}
In order to derive a lower bound on the numerator of the right-hand side of
\eqref{eq: fourth step in the derivation of the improved lower bound on total variation distance},
for the function $f$ in \eqref{eq: proposed f for the derivation of the improved lower bound on the
total variation distance}, we need to calculate the expected value of the right-hand side of
\eqref{eq: 11th step in the derivation of the improved lower bound on total variation distance}.
To this end, the first and second moments of $U_j$ are calculated as follows:
\begin{eqnarray}
&& \expectation(U_j) \nonumber \\
&& = \expectation(V_j) - \lambda \nonumber \\
&& = \sum_{i \neq j} p_i - \sum_{i=1}^n p_i \nonumber \\
&& = -p_j \label{eq: first moment of U_j} \\
\text{and} \nonumber \\
&& \expectation(U_j^2) \nonumber \\
&& = \text{Var}(U_j) + \bigl(\expectation(U_j)\bigr)^2 \nonumber \\
&& = \text{Var}(V_j) + p_j^2 \nonumber \\
&& \stackrel{(\text{a})}{=} \sum_{i \neq j} p_i(1-p_i) + p_j^2 \nonumber \\
&& = \sum_{i \neq j} p_i - \sum_{i \neq j} p_i^2 + p_j^2 \nonumber \\
&& = \lambda - p_j - \sum_{i \neq j} p_i^2 + p_j^2.
\label{eq: second moment of U_j}
\end{eqnarray}
where equality~(a) holds since the binary random variables $\{X_i\}_{i=1}^n$
are independent and
$\text{Var}(X_i) = p_i (1-p_i)$. By taking expectations on both sides of
\eqref{eq: 11th step in the derivation of the improved lower bound on total variation distance},
one obtains from \eqref{eq: first moment of U_j} and \eqref{eq: second moment of U_j} that
\begin{eqnarray}
&& \expectation\bigl[f(V_j+2) - f(V_j+1)\bigr] \nonumber \\[0.2cm]
&& \geq 1 - \frac{3 \Bigl(\lambda - p_j - \sum_{i \neq j} p_i^2 + p_j^2\Bigr)
+ 3 \bigl(3-2\alpha_2 + 2 \lambda\bigr) \bigl(-p_j \bigr) +
(2-\alpha_2+\lambda)^3 - (1-\alpha_2+\lambda)^3}{\theta \lambda} \nonumber \\[0.2cm]
&& \hspace*{0.4cm} - \bigl|\alpha_2 - \alpha_1\bigr| \cdot
\exp\left(-\frac{(1-\alpha_2)_+^2}{\theta \lambda} \right) \cdot
\left(\frac{-2p_j + 2\lambda + \left|3 - 2 \alpha_2 \right|}{\theta \lambda} \right)
\nonumber \\[0.2cm]
&& = 1 - \frac{3 \lambda + (2-\alpha_2+\lambda)^3 - (1-\alpha_2+\lambda)^3
- \Bigl[3p_j(1-p_j) + 3 \sum_{i \neq j} p_i^2
+ 3 \bigl(3-2\alpha_2 + 2 \lambda\bigr) p_j \Bigr]}{\theta \lambda} \nonumber \\[0.2cm]
&& \hspace*{0.4cm} - \biggl(\frac{\bigl|\alpha_2 - \alpha_1\bigr| \,
\bigl(2\lambda - 2 p_j+ \left|3 - 2 \alpha_2 \right| \bigr)}{\theta \lambda} \biggr)
\cdot \exp\left(-\frac{(1-\alpha_2)_+^2}{\theta \lambda} \right) \nonumber \\[0.2cm]
&& \geq 1 - \frac{3 \lambda + (2-\alpha_2+\lambda)^3 - (1-\alpha_2+\lambda)^3
- \bigl(9-6\alpha_2 + 6 \lambda\bigr) p_j}{\theta \lambda} \nonumber \\[0.2cm]
&& \hspace*{0.4cm} - \biggl(\frac{\bigl|\alpha_2 - \alpha_1\bigr| \,
\bigl(2\lambda+ \left|3 - 2 \alpha_2 \right| \bigr)}{\theta \lambda} \biggr)
\cdot \exp\left(-\frac{(1-\alpha_2)_+^2}{\theta \lambda} \right) \; .
\label{eq: 12th step in the derivation of the improved lower bound on total variation distance}
\end{eqnarray}
Therefore, from \eqref{eq: 12th step in the derivation of the improved lower bound on total variation distance},
the following lower bound on the right-hand side of
\eqref{eq: fourth step in the derivation of the improved lower bound on total variation distance}
holds
\begin{eqnarray}
&& \hspace*{-1.5cm} \sum_{j=1}^n \Bigl\{ p_j^2 \,
\expectation \bigl[f(V_j+2) - f(V_j+1)\bigr] \Bigr\} \geq
\left( \frac{3\bigl(3-2\alpha_2+2\lambda\bigr)}{\theta \lambda} \right)
\sum_{j=1}^n p_j^3 \nonumber \\
&& \hspace*{-1.2cm} + \left(1-\frac{3 \lambda + (2-\alpha_2+\lambda)^3 - (1-\alpha_2+\lambda)^3
+ |\alpha_1 - \alpha_2| \bigl(2\lambda+|3-2\alpha_2|\bigr)
\exp\left(-\frac{(1-\alpha_2)_+^2}{\theta \lambda} \right)}{\theta \lambda} \right)
\sum_{j=1}^n p_j^2 \, .
\label{eq: 12.5th step in the derivation of the improved lower bound on total variation distance}
\end{eqnarray}
Note that if $\alpha_2 \leq \lambda + \frac{3}{2}$, which is a condition that
is involved in the maximization of
\eqref{eq: K1 in the lower bound on the total variation distance}, then the first term on
the right-hand side of
\eqref{eq: 12.5th step in the derivation of the improved lower bound on total variation distance}
can be removed, and the resulting lower bound on the numerator of the right-hand side of
\eqref{eq: fourth step in the derivation of the improved lower bound on total variation distance}
takes the form
\begin{equation}
\sum_{j=1}^n \Bigl\{ p_j^2 \, \expectation\bigl[f(V_j+2) - f(V_j+1)\bigr] \Bigr\}
\geq \Bigl(1 - h_{\lambda}(\alpha_1, \alpha_2, \theta) \Bigr) \sum_{j=1}^n p_j^2
\label{eq: 13th step in the derivation of the improved lower bound on total variation distance}
\end{equation}
where the function $h_{\lambda}$ is introduced in \eqref{eq: h in the lower bound on the total
variation distance}.

We turn now to deriving an upper bound on the denominator of the right-hand side of
\eqref{eq: fourth step in the derivation of the improved lower bound on total variation distance}.
Therefore, we need to derive a closed-form upper bound on
$\sup_{k \in \naturals_0} \bigl| \lambda \, f(k+1) - k \, f(k) \bigr|$
with the function $f$ in \eqref{eq: proposed f for the derivation of the improved lower
bound on the total variation distance}. For every $k \in \naturals_0$,
\begin{equation}
\lambda \, f(k+1) - k \, f(k) = \lambda \, \bigl[ f(k+1) - f(k) \bigr] + (\lambda - k) \, f(k).
\label{eq: 14th step in the derivation of the improved lower bound on total variation distance}
\end{equation}
In the following, we derive bounds on each of the two terms on the right-hand side of
\eqref{eq: 14th step in the derivation of the improved lower bound on total variation distance},
and we start with the first term.
Let $$t(u) \triangleq (u+\alpha_2-\alpha_1) \, \exp\left(-\frac{u^2}{\theta \lambda}\right),
\quad \forall \, u \in \reals$$
then $f(k) = t(k-\alpha_2)$ for every $k \in \naturals_0$, and by the mean-value theorem of calculus,
\begin{eqnarray}
&& f(k+1) - f(k) \nonumber \\
&& = t(k+1-\alpha_2) - t(k-\alpha_2) \nonumber \\
&& = t'(c_k) \quad \mbox{for some} \; c_k \in [k-\alpha_2, k+1-\alpha_2] \nonumber \\
&& = \left(1-\frac{2 c_k^2}{\theta \lambda}\right) \, \exp\left(-\frac{c_k^2}{\theta
\lambda}\right) + \left(\frac{2(\alpha_1-\alpha_2) c_k}{\theta \lambda}\right) \,
\exp\left(-\frac{c_k^2}{\theta \lambda}\right) \, .
\label{eq: 15th step in the derivation of the improved lower bound on total variation distance}
\end{eqnarray}
Referring to the first term on the right-hand side of
\eqref{eq: 15th step in the derivation of the improved lower bound on total variation distance},
let $$p(u) \triangleq (1-2u) e^{-u}, \quad \forall \, u \geq 0$$ then the global maximum and
minimum of $p$ over the non-negative real line are obtained at
$u=0$ and $u = \frac{3}{2}$, respectively, and therefore
$$ -2 e^{-\frac{3}{2}} \leq p(u) \leq 1, \quad \forall \, u \geq 0.$$
Let $u = \frac{c_k^2}{\theta \lambda}$; then it follows that the first term on the right-hand side
of \eqref{eq: 15th step in the derivation of the improved lower bound on total variation distance}
satisfies the inequality
\begin{equation}
-2 e^{-\frac{3}{2}} \leq \Bigl(1-\frac{2 c_k^2}{\theta \lambda}\Bigr) \,
\exp\Bigl(-\frac{c_k^2}{\theta \lambda}\Bigr) \leq 1.
\label{eq: 16th step in the derivation of the improved lower bound on total variation distance}
\end{equation}
Furthermore, referring to the second term on the right-hand side of
\eqref{eq: 15th step in the derivation of the improved lower bound on total variation distance},
let $$q(u) \triangleq u e^{-u^2}, \quad \forall \, u \in \reals$$
then the global maximum and minimum of $q$ over the real line are obtained at
$u = +\frac{\sqrt{2}}{2}$ and $u = -\frac{\sqrt{2}}{2}$, respectively, and therefore
$$-\frac{1}{2} \sqrt{\frac{2}{e}} \leq q(u) \leq +\frac{1}{2} \sqrt{\frac{2}{e}} \; ,
\quad \forall \, u \in \reals.$$
Let this time $u = \sqrt{\frac{c_k}{\theta \lambda}}$; then it follows that the second term on the
right-hand side
of \eqref{eq: 15th step in the derivation of the improved lower bound on total variation distance}
satisfies
\begin{equation}
\left| \biggl(\frac{2(\alpha_1 - \alpha_2) c_k}{\theta \lambda} \biggr) \cdot
\exp\biggl(-\frac{c_k^2}{\theta \lambda}\biggr) \right| \leq
\sqrt{\frac{2}{\theta \lambda e}} \cdot |\alpha_1 - \alpha_2|.
\label{eq: 17th step in the derivation of the improved lower bound on total variation distance}
\end{equation}
Hence, on combining the equality in
\eqref{eq: 15th step in the derivation of the improved lower bound on total variation distance}
with the two inequalities in
\eqref{eq: 16th step in the derivation of the improved lower bound on total variation distance}
and
\eqref{eq: 17th step in the derivation of the improved lower bound on total variation distance},
it follows that the first term on the right-hand side of
\eqref{eq: 14th step in the derivation of the improved lower bound on total variation distance}
satisfies
\begin{equation}
-\left(2 \lambda e^{-\frac{3}{2}} + \sqrt{\frac{2\lambda}{\theta e}} \cdot |\alpha_1 -
\alpha_2| \right) \leq \lambda \bigl[f(k+1) - f(k) \bigr]
\leq \lambda + \sqrt{\frac{2\lambda}{\theta e}} \cdot |\alpha_1 - \alpha_2| \, ,
\quad \forall \, k \in \naturals_0.
\label{eq: 18th step in the derivation of the improved lower bound on total variation distance}
\end{equation}
We continue the analysis by a derivation of bounds on the second term of the right-hand side
of \eqref{eq: 14th step in the derivation of the improved lower bound on total variation distance}.
For the function $f$ in \eqref{eq: proposed f for the derivation of the improved lower bound on
the total variation distance}, it is equal to
\begin{eqnarray}
&& (\lambda - k) \, f(k) \nonumber \\
&& = (\lambda-k) (k-\alpha_1) \exp\biggl(-\frac{(k-\alpha_2)^2}{\theta \lambda}\biggr) \nonumber \\
&& = \bigl[ (\lambda-\alpha_2)+(\alpha_2-k) \bigr] \, \bigl[ (k-\alpha_2)+(\alpha_2-\alpha_1)
\bigr] \, \exp\biggl(-\frac{(k-\alpha_2)^2}{\theta \lambda}\biggr) \nonumber \\
&& = \Bigl[ (\lambda-\alpha_2) (k-\alpha_2) +
(\alpha_2-\alpha_1) (\lambda-\alpha_2) - (k-\alpha_2)^2
+ (\alpha_1 - \alpha_2) (k-\alpha_2) \Bigr] \, \exp\biggl(-\frac{(k-\alpha_2)^2}{\theta
\lambda}\biggr) \nonumber \\[0.1cm]
&& = \bigl[ \sqrt{\theta \lambda} \, (\lambda-\alpha_2) \, v_k - \theta \lambda \, v_k^2
- \sqrt{\theta \lambda} \, (\alpha_2 - \alpha_1) \, v_k
+ (\alpha_2 - \alpha_1) (\lambda - \alpha_2) \bigr] \, e^{-v_k^2} \, , \quad
v_k \triangleq \frac{k-\alpha_2}{\sqrt{\theta \lambda}} \; \; \forall \, k \in \naturals_0
\nonumber \\
&& = (c_0 + c_1 v_k + c_2 v_k^2) \, e^{-v_k^2}
\label{eq: 19th step in the derivation of the improved lower bound on total variation distance}
\end{eqnarray}
where the coefficients $c_0, c_1$ and $c_2$ are introduced in
Eqs.~\eqref{eq: c0}--\eqref{eq: c2}, respectively.
In order to derive bounds on the left-hand side of
\eqref{eq: 19th step in the derivation of the improved lower bound on total variation distance},
let us find the global maximum and minimum of the function $x$
in \eqref{eq: function x}:
$$x(u) \triangleq (c_0 + c_1 u + c_2 u^2) e^{-u^2} \, \quad \forall \, u \in \reals.$$
Note that $\lim_{u \rightarrow \pm \infty} x(u) = 0$ and $x$ is differentiable
over the real line, so the global maximum and minimum of $x$ are attained at
finite points and their corresponding values are finite. By setting the derivative of $x$ to zero,
we have that the candidates for the global maximum and minimum of $x$ over the real line are the real zeros
$\{u_i\}$ of the cubic polynomial equation in \eqref{eq: zeros of a cubic polynomial equation}.
Note that by their definition in \eqref{eq: zeros of a cubic polynomial equation}, the values
of $\{u_i\}$ are {\em independent} of the value of $k \in \naturals_0$, and also the
size of the set $\{u_i\}$ is equal to~3 (see Remark~\ref{remark: the size of the set of real zeros
is equal to 3}). Hence, it follows from
\eqref{eq: 19th step in the derivation of the improved lower bound on total variation distance}
that
\begin{equation}
\min_{i \in \{1, 2, 3\}} \{x(u_i)\} \leq (\lambda - k) \, f(k)
\leq \max_{i \in \{1, 2, 3\}} \{x(u_i)\} \, , \quad  \forall \, k \in \naturals_0
\label{eq: 20th step in the derivation of the improved lower bound on total variation distance}
\end{equation}
where these bounds on the second term on the right-hand side of
\eqref{eq: 14th step in the derivation of the improved lower bound on total variation distance}
are independent of the value of $k \in \naturals_0$.

In order to get bounds on the left-hand side of
\eqref{eq: 14th step in the derivation of the improved lower bound on total variation distance},
note that from the bounds on the first and second terms on the right-hand side of
\eqref{eq: 14th step in the derivation of the improved lower bound on total variation distance}
(see
\eqref{eq: 18th step in the derivation of the improved lower bound on total variation distance}
and
\eqref{eq: 20th step in the derivation of the improved lower bound on total variation distance},
respectively) then for every $k \in \naturals_0$
\begin{eqnarray}
&& \min_{i \in \{1, 2, 3\}} \{x(u_i)\} - \left(2 \lambda e^{-\frac{3}{2}} +
\sqrt{\frac{2\lambda}{\theta e}} \cdot |\alpha_1 - \alpha_2| \right) \nonumber \\[0.1cm]
&& \leq \lambda \, f(k+1) - k \, f(k) \nonumber \\[0.1cm]
&& \leq \max_{i \in \{1, 2, 3\}} \{x(u_i)\}
+ \lambda + \sqrt{\frac{2\lambda}{\theta e}} \cdot |\alpha_1 - \alpha_2|
\label{eq: 21st step in the derivation of the improved lower bound on total variation distance}
\end{eqnarray}
which yields that the following inequality is satisfied:
\begin{equation}
\sup_{k \in \naturals_0} \left| \lambda \, f(k+1) - k \, f(k) \right|
\leq g_{\lambda}(\alpha_1, \alpha_2, \theta)
\label{eq: 22nd step in the derivation of the improved lower bound on total variation distance}
\end{equation}
where the function $g_{\lambda}$ is introduced in
\eqref{eq: g in the lower bound on the total variation distance}.
Finally, by combining the inequalities in
Eqs.~\eqref{eq: fourth step in the derivation of the improved lower bound on total variation distance},
\eqref{eq: 13th step in the derivation of the improved lower bound on total variation distance}
and
\eqref{eq: 22nd step in the derivation of the improved lower bound on total variation distance},
the lower bound on the total variation distance in
\eqref{eq: improved lower bound on the total variation distance} follows. The existing
upper bound on the total variation distance in
\eqref{eq: improved lower bound on the total variation distance}
was derived in Theorem~1 of Barbour and Hall (1984) (see
Theorem~\ref{theorem: bounds on the total variation distance - Barbour and Hall 1984} here).
This completes the proof of
Theorem~\ref{theorem: improved lower bound on the total variation distance}.

\vspace*{0.1cm}
\subsection{Proof of Corollary~\ref{corollary: lower bound on the total variation distance}}
\label{subsubsection: Proof of the corollary with the improved lower bound on the total
variation distance}
Corollary~\ref{corollary: lower bound on the total variation distance} follows as a special
case of Theorem~\ref{theorem: improved lower bound on the total variation distance} when
the proposed function $f$ in \eqref{eq: proposed f for the derivation of the improved lower
bound on the total variation distance} is chosen such that two of its three free
parameters (i.e., $\alpha_1$ and $\alpha_2$) are determined sub-optimally, and its third
parameter ($\theta$) is determined optimally in terms of the sub-optimal selection of
the two other parameters. More explicitly, let $\alpha_1$ and $\alpha_2$ in
\eqref{eq: proposed f for the derivation of the improved lower bound on the total variation
distance} be set to be equal to $\lambda$ (i.e., $\alpha_1 = \alpha_2 = \lambda$). From
\eqref{eq: c0}--\eqref{eq: c2}, this setting implies that $c_0 = c_1 = 0$ and
$c_2 = -\theta \lambda < 0$ (since $\theta, \lambda > 0$). The cubic polynomial
equation in \eqref{eq: zeros of a cubic polynomial equation}, which corresponds to this
(possibly sub-optimal) setting of $\alpha_1$ and $\alpha_2$, is $$2 c_2 u^3  - 2 c_2 u = 0$$
whose zeros are $u = 0, \pm 1$. The
function $x$ in \eqref{eq: function x} therefore takes the form
$$x(u) = c_2 u^2 e^{-u^2} \, \quad \forall \, u \in \reals$$
so $x(0)=0$ and $x(\pm 1) = \frac{c_2}{e} < 0$. This implies that
$$\min_{i \in \{1, 2, 3\}} x(u_i) = \frac{c_2}{e}, \quad
\max_{i \in \{1, 2, 3\}} x(u_i) = 0,$$ and therefore $h_{\lambda}$ and
$g_{\lambda}$ in \eqref{eq: h in the lower bound on the total variation distance} and
\eqref{eq: g in the lower bound on the total variation distance}, respectively,
are simplified to
\begin{eqnarray}
&& h_{\lambda}(\lambda, \lambda, \theta) = \frac{3 \lambda+7}{\theta \lambda} \, ,
\label{eq: simplified h for the corollary on the total variation distance} \\
&& g_{\lambda}(\lambda, \lambda, \theta)
= \lambda \, \max\bigl\{1, 2 e^{-\frac{3}{2}} + \theta e^{-1} \bigr\}.
\label{eq: simplified g for the corollary on the total variation distance}
\end{eqnarray}
This sub-optimal setting of $\alpha_1$ and $\alpha_2$ in \eqref{eq: proposed f for the
derivation of the improved lower bound on the total variation distance} implies that the
coefficient $K_1$ in \eqref{eq: K1 in the lower bound on the total variation distance}
is replaced with a loosened version
\begin{equation}
K'_1(\lambda) \triangleq \sup_{\theta > 0} \left(\frac{1-h_{\lambda}(\lambda,
\lambda, \theta)}{2 g_{\lambda}(\lambda, \lambda, \theta)}\right).
\label{eq: loosened coefficient for the corollary on the total variation distance}
\end{equation}
Let $\theta \geq e - \frac{2}{\sqrt{e}}$; then
\eqref{eq: simplified g for the corollary on the total variation distance} is simplified
to $g_{\lambda}(\lambda, \lambda, \theta)
= \lambda \, \bigl(2 e^{-\frac{3}{2}} + \theta e^{-1}\bigr)$. It therefore
follows from \eqref{eq: improved lower bound on the total variation distance},
\eqref{eq: K1 in the lower bound on the total variation distance}
and  \eqref{eq: simplified h for the corollary on the total variation distance}--\eqref{eq:
loosened coefficient for the corollary on the total variation distance} that
\begin{equation}
d_{\text{TV}}\bigl(P_W, \text{Po}(\lambda)\bigr) \geq \widetilde{K}_1(\lambda) \,
\sum_{i=1}^n p_i^2
\label{eq: loosened lower bound on the total variation distance - Corollary}
\end{equation}
where
\begin{equation}
\widetilde{K}_1(\lambda)
= \sup_{\theta \geq e - \frac{2}{\sqrt{e}}} \, \left( \frac{1-\frac{3\lambda+7}{\theta
\lambda}}{2\lambda \bigl(2 e^{-\frac{3}{2}} + \theta e^{-1}\bigr)} \right)
\label{eq: a possibly further loosened coefficient for the corollary on the total variation distance}
\end{equation}
and, in general, $K'_1(\lambda) \geq \widetilde{K}_1(\lambda)$ due to the
above restricted constraint on $\theta$ (see \eqref{eq: loosened coefficient
for the corollary on the total variation distance} versus
\eqref{eq: a possibly further loosened coefficient for the corollary on the
total variation distance}). Differentiating the function inside the
supremum w.r.t. $\theta$ and setting its derivative
to zero, one gets the following quadratic equation in $\theta$:
$$\lambda \, \theta^2 - 2(3\lambda+7) \, \theta - 2(3\lambda+7) e^{-1} = 0$$
whose positive solution is the optimized value of $\theta$
in \eqref{eq: optimal theta for alpha1 and alpha2 equal to lambda}.
Furthermore, it is clear that this value of $\theta$ in
\eqref{eq: optimal theta for alpha1 and alpha2 equal to lambda}
is larger than, e.g., 3, so it satisfies the constraint
in \eqref{eq: a possibly further loosened coefficient
for the corollary on the total variation distance}. This completes
the proof of Corollary~\ref{corollary: lower bound on the total variation distance}.

\section{A Comparison of the New Bounds with Known Results}
\label{section: Connection of the improved lower bounds with known results}

The new lower bounds on the total variation distance in
Theorem~\ref{theorem: improved lower bound on the total variation distance} and
Corollary~\ref{corollary: lower bound on the total
variation distance}
scale like $\sum_{i=1}^n p_i^2$, similarly to the known upper and lower
bounds in Theorem~\ref{theorem: bounds on the total variation distance -
Barbour and Hall 1984} that originally appear in Theorems~1 and 2 of Barbour and Hall (1984).
However, the new lower bounds offer a significant
improvement over the known lower bound in
Theorem~\ref{theorem: bounds on the total variation distance -
Barbour and Hall 1984}. More explicitly,
from Theorems~1 and 2 of Barbour and Hall (1984), the ratio between
the upper and lower bounds on the total variation distance (see
\eqref{eq: bounds on the total variation distance - Barbour and Hall 1984}) is equal to 32
in the two extreme cases where $\lambda \rightarrow 0$ or $\lambda \rightarrow \infty$. In
the following, we calculate the ratio of the same upper bound and the new lower bound in
Corollary~\ref{corollary: lower bound on the total variation distance} at these two
extreme cases. In the limit where $\lambda \rightarrow \infty$, this
ratio tends to
\begin{eqnarray}
&& \lim_{\lambda \rightarrow \infty} \frac{\left(\frac{1-e^{-\lambda}}{\lambda}\right)
\, \dsum_{i=1}^n p_i^2}{\left(\frac{1-\frac{3\lambda+7}{\lambda \theta}}{2 \lambda
\bigl(2 e^{-{3/2}}+\theta \, e^{-1}\bigr)} \right)\, \dsum_{i=1}^n p_i^2} \quad
\quad (\theta = \theta(\lambda) \;
\mbox{is given in Eq.~\eqref{eq: optimal theta for alpha1 and alpha2 equal to lambda}})
\nonumber \\[0.1cm]
&& = \frac{2}{e} \, \lim_{\lambda \rightarrow \infty} \frac{\theta \bigl(2 e^{-{1/2}}+\theta
\bigr)}{\theta-\bigl(3+\frac{7}{\lambda}\bigr)} \nonumber \\[0.1cm]
&& = \frac{6}{e} \, \left(1+\sqrt{1+\frac{2}{3} \cdot e^{-1/2}} \right)^2 \approx 10.539
\label{eq: limit of the ratio between the upper and lower bounds on the total variation
distance when lambda tends to infinity}
\end{eqnarray}
where the last equality follows from \eqref{eq: optimal theta for alpha1 and alpha2 equal to lambda}, since
$\lim_{\lambda \rightarrow \infty} \theta = 3+\sqrt{3(3+2e^{-1/2})}$.
Furthermore, the limit of this ratio when $\lambda \rightarrow 0$ is equal to
\begin{eqnarray}
&& 2 \, \lim_{\lambda \rightarrow 0} \left(\frac{1-e^{-\lambda}}{\lambda} \right) \,
\lim_{\lambda \rightarrow 0} \left(\frac{\lambda \bigl(2 e^{-{3/2}}+\theta \, e^{-1} \bigr)}{1-
\frac{3\lambda+7}{\lambda \theta}}\right) \nonumber \\[0.1cm]
&& \stackrel{\text{(a)}}{=} \frac{28}{e} \, \lim_{\lambda \rightarrow 0}
\left(\frac{2 e^{-{1/2}}+\theta)}{\theta-\bigl(3+\frac{7}{\lambda}\bigr)}\right) \nonumber \\
&& \stackrel{\text{(b)}}{=} \frac{56}{e} \approx 20.601
\label{eq: limit of the ratio between the upper and lower bounds on the total variation
distance when lambda tends to zero}
\end{eqnarray}
where equalities~(a) and (b) hold since, from
\eqref{eq: optimal theta for alpha1 and alpha2 equal to lambda}, it follows that
$\lim_{\lambda \rightarrow 0} (\lambda \theta) = 14$. This
implies that Corollary~\ref{corollary: lower bound on the total variation distance} improves
the original lower bound on the total variation distance in Theorem~2 of Barbour and Hall (1984)
by a factor of $\frac{32}{10.539} \approx 3.037$ in the limit where $\lambda \rightarrow \infty$,
and it also improves it by a factor of $\frac{32}{20.601} \approx 1.553$ if
$\lambda \rightarrow 0$ while still having a closed-form expression for the lower bound in
Corollary~\ref{corollary: lower bound on the total variation distance}. The only reason
for this improvement is related to the optimal choice of the free parameter $\theta$
in \eqref{eq: optimal theta for alpha1 and alpha2 equal to lambda},
versus its sub-optimal choice in the proof of Theorem~2 of Barbour and Hall (1984). This observation
has motivated to further improve the lower bound by introducing the
two additional parameters $\alpha_1, \alpha_2 \in \reals$ in
Theorem~\ref{theorem: improved lower bound on the total variation distance}; these parameters
give two additional degrees of freedom in the function $f$
in \eqref{eq: proposed f for the derivation of the improved lower bound on the total variation distance}
(according to the proof in
Section~\ref{subsubsection: Proof of the corollary with the improved lower bound on the total variation distance},
these two parameters are set to be equal to $\lambda$ for the derivation of the loosened and simplified bound in
Corollary~\ref{corollary: lower bound on the total variation distance}).
The improvement in the lower bound of
Theorem~\ref{theorem: improved lower bound on the total variation distance} (in comparison to
Corollary~\ref{corollary: lower bound on the total variation distance})
is especially dominant for low values of $\lambda$, as is shown in
Figure~\ref{Figure: ratio ot upper and lower bounds on the total variation distance}. Note, however, that no improvement is obtained for high values of $\lambda$ (e.g., for $\lambda \geq 20$, as is shown
by Figure~\ref{Figure: ratio ot upper and lower bounds on the total variation distance} on noticing that
the curves in this plot merge at large values of $\lambda$).
\begin{figure}[here!]  
\begin{center}
\epsfig{file=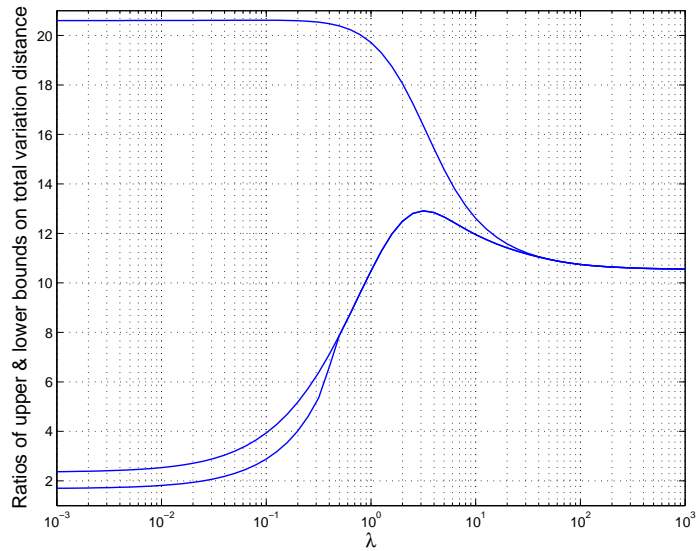,scale=0.53}
\end{center}
\caption{\label{Figure: ratio ot upper and lower bounds on the total variation distance} The
figure presents curves that correspond to ratios of upper and lower bounds on the total
variation distance between the sum of independent Bernoulli random variables and the Poisson
distribution with the same mean $\lambda$. The upper bound on the total variation distance
for all these three curves is the bound given by Barbour and Hall (see Theorem~1 of Barbour and Hall (1984)
or Theorem~\ref{theorem: bounds on the total variation distance - Barbour and Hall 1984}
here). The lower bounds that the three curves refer to are the following.
The curve at the bottom (i.e., the one which provides the lowest ratio for a fixed
$\lambda$) is the improved lower bound on the total variation distance that is
introduced in Theorem~\ref{theorem: improved lower bound on the total variation distance}.
The curve slightly above it for small values of $\lambda$ corresponds to the looser lower bound
obtained when $\alpha_1$ and $\alpha_2$ in
\eqref{eq: K1 in the lower bound on the total variation distance} are set to be equal
(i.e., $\alpha_1 = \alpha_2 \triangleq \alpha$ is their common value), and so the
optimization of $K_1$ for this curve is reduced to a two-parameter maximization of $K_1$
over the two free parameters $\alpha \in \reals$ and $\theta \in \reals^+$. Finally,
the curve at the top of this figure corresponds to the further loosening of this lower bound
where $\alpha$ is set to be equal to $\lambda$; this leads to a single-parameter maximization
of $K_1$ (over the parameter $\theta \in \reals^+$) whose optimization leads to the closed-form
expression for the lower bound in
Corollary~\ref{corollary: lower bound on the total variation distance}. For comparison,
in order to assess the enhanced tightness of the new lower bounds, note that the ratio of the
upper and lower bounds on the total variation distance from
Theorems~1 and 2 of Barbour and Hall (1984) (or
Theorem~\ref{theorem: bounds on the total variation distance - Barbour and Hall 1984} here)
is roughly equal to 32 for all values of $\lambda$.}
\end{figure}

The lower bound on the total variation distance in
Theorem~\ref{theorem: improved lower bound on the total variation distance} implies the bound in
Corollary~\ref{corollary: lower bound on the total variation distance} (see the proof in
Section~\ref{subsubsection: Proof of the corollary with the improved lower bound on the total variation distance}).
Corollary~\ref{corollary: lower bound on the total variation distance}
further implies the lower bound on the total variation distance in Theorem~2 of Barbour and Hall (1984) (see
Theorem~\ref{theorem: bounds on the total variation distance - Barbour and Hall 1984} here).
The latter claim follows from the
fact that the lower bound in \eqref{eq: loosened lower bound on the total variation distance - Corollary}
with the coefficient
$\widetilde{K}_1(\lambda)$ in \eqref{eq: a possibly further loosened coefficient for the corollary
on the total variation distance} was loosened in the proof of Theorem~2 of Barbour and Hall (1984) by a
sub-optimal selection of the parameter $\theta$, which leads to a lower bound on $\widetilde{K}_1(\lambda)$
(the sub-optimal selection of $\theta$ in the proof of
Theorem~2 of Barbour and Hall (1984) is $\theta = 21 \max\bigl\{1, \frac{1}{\lambda}\bigr\}$).
On the other hand, the optimized value of $\theta$ that is used
in \eqref{eq: optimal theta for alpha1 and alpha2 equal to lambda} provides an exact
closed-form expression for $\widetilde{K}_1(\lambda)$ in
\eqref{eq: a possibly further loosened coefficient for the corollary on the total variation distance},
and it leads to the derivation of the improved lower bound in
Corollary~\ref{corollary: lower bound on the total variation distance}.

Theorem~1.2 of Deheuvels and Pfeifer (1986) provides an asymptotic result for the
total variation distance between the distribution of the sum $W$ of $n$ independent
Bernoulli random variables with $\expectation(X_i) = p_i$ and the Poisson
distribution with mean $\lambda = \sum_{i=1}^n p_i$. It shows that when
$\sum_{i=1}^n p_i \rightarrow \infty$ and
$\max_{1 \leq i \leq n} p_i \rightarrow 0$ as $n \rightarrow \infty$ then
\begin{equation}
d_{\text{TV}}(P_W, \text{Po}(\lambda)) \sim \frac{1}{\sqrt{2 \pi e} \;
\lambda} \; \sum_{i=1}^n p_i^2 \, .
\label{eq: asymptotic expression for the total variation distance when lambda
tends to infinity and p_max to zero}
\end{equation}
This implies that the ratio of the upper bound on the total variation distance
in Theorem~1 of Barbour and Hall (1984) (see Theorems~\ref{theorem: bounds on the
total variation distance - Barbour and Hall 1984} here) and this asymptotic
expression is equal to $\sqrt{2 \pi e} \approx 4.133$.
Therefore, the ratio between the exact
asymptotic value in \eqref{eq: asymptotic expression for the total variation
distance when lambda tends to infinity and p_max to zero} and the new lower
bound in \eqref{eq: improved lower bound on the total variation distance}
is equal to $\frac{10.539}{\sqrt{2 \pi e}} \approx 2.55$. It therefore follows
that, in the limit where $\lambda \rightarrow 0$,
the new lower bound on the total variation in
\eqref{eq: improved lower bound on the total variation distance} is smaller
than the exact value by no more than 1.69, and for $\lambda \gg 1$, it is
smaller than the exact asymptotic result by a factor of 2.55.

\vspace*{0.2cm}
{\em Acknowledgment}:
The anonymous reviewer is acknowledged for suggestions that
led to improvement of the presentation of this paper.
This research work was supported by the Israeli Science Foundation (ISF), grant number 12/12.

\section*{References}
{\footnotesize \hspace*{-0.45cm}
Arratia, R., Goldstein, L., Gordon, L., 1990. Poisson approximation
and the Chen-Stein method. Statistical Science~5, 403--424.\\
Barbour, A.D., Chen, L. H. Y., 2005. An Introduction to Stein's
Method. Lecture Notes Series, Institute for Mathematical Sciences,
Singapore University Press and World Scientific.\\
Barbour, A.D., Hall, P., 1984. On the rate of Poisson convergence.
Mathematical Proceedings of the Cambridge Philosophical Society~95, 473--480.\\
Barbour, A.D., Holst, L., Janson, S., 1992. Poisson Approximation.
Oxford University Press.\\
Barbour, A. D., Johnson, O., Kontoyiannis, I., Madiman, M., 2010.
Compound Poisson approximation via information functionals.
Electronic Journal of Probability~15, 1344--1369.\\
Chen, L.H.Y., 1975. Poisson approximation for dependent trials.
Annals of Probability~3, 534--545.\\
Deheuvels, P., Pfeifer, D., 1986. A semigroup approach to Poisson
approximation. Annals of Probability~14, 663--676.\\
Kontoyiannis, I., Harremo{\"{e}}s, P., Johnson, O., 2005. Entropy and the
law of small numbers. IEEE Trans. on Information Theory~51, 466--472.\\
Le Cam, L., 1960. An approximation theorem for the Poisson binomial
distribution. Pacific Journal of Mathematics~10, 1181--1197.\\
Ross, S.M., Pek{\"{o}}z, E.A. 2007. A Second Course in Probability.
Probability Bookstore.\\
Ross, N., 2011. Fundamentals of Stein's Method. Probability Surveys~8,
210--293.\\
Sason, I., 2012. Entropy bounds for discrete random variables via coupling.
[Online]. Available: \url{http://arxiv.org/abs/1209.5259}.
}
\end{document}